\documentclass[12pt,onecolumn,draftcls,journal,letterpaper]{IEEEtran}

\usepackage{amsthm} 
\usepackage{amsfonts,amssymb,amstext,verbatim,amsmath,graphicx,color,enumerate}

\usepackage[normalem]{ulem}

\usepackage{tikz}

\newtheorem{assumption}{\textbf{Assumption}}
\newtheorem{theorem}{\textbf{Theorem}}[section]

\newtheorem{lemma}[theorem]{\textbf{Lemma}}

\newtheorem{remark}[theorem]{\textbf{Remark}}
\newtheorem{proposition}[theorem]{\textbf{Proposition}}
\newtheorem{definition}[theorem]{\textbf{Definition}}

\newcommand{\argmax}{\mathop{\operatorname{argmax}}}

\long\def\symbolfootnote[#1]#2{\begingroup%
\def\thefootnote{\fnsymbol{footnote}}\footnote[#1]{#2}\endgroup}

\newcommand{\EE}{\mathcal{E}}
\newcommand{\GG}{\mathcal{G}}
\newcommand{\LL}{\mathcal{L}}

\newcommand{\nbrs}{\mathcal{N}}

\newcommand{\until}[1]{\{1,\ldots,#1\}} 
\newcommand{\subj}{\text{subj. to}}
\newcommand\oprocendsymbol{\hbox{$\square$}}
\newcommand\oprocend{\relax\ifmmode\else\unskip\hfill\fi\oprocendsymbol}
\def\eqoprocend{\tag*{$\square$}}
\newcommand{\real}{{\mathbb{R}}}
\newcommand{\1}{\mathbf{1}}
\newcommand{\map}[3]{#1: #2 \rightarrow #3}

\newcommand{\sx}[1]{x^{#1}}

\newcommand{\smu}[1]{\mu^{#1}}
\newcommand{\slambda}[1]{{\lambda^{#1}}}
\newcommand{\srho}[1]{\rho^{#1}}
\newcommand{\slotIndex}{s}
\newcommand{\slotUB}{S}

\def \Dminmax/{Primal Min-Max Dual Subgradient}

\renewcommand{\inf}{\operatornamewithlimits{inf\vphantom{p}}}

\renewcommand{\lim}{\operatornamewithlimits{lim\vphantom{p}}}

\usepackage{algorithm}
\usepackage[noend]{algpseudocode}
\makeatletter
\newcommand{\StatexIndent}[1][3]{%
  \setlength\@tempdima{\algorithmicindent}%
  \Statex\hskip\dimexpr#1\@tempdima\relax}
\makeatother

\graphicspath{{figs/}}

\begin{document}

\title{
  A duality-based approach for distributed min-max optimization\\
  with application to demand side management
  }

\author{Ivano Notarnicola$^1$, Mauro Franceschelli$^2$, Giuseppe Notarstefano$^1$
\thanks{The research leading to these results has received funding from the
  European Research Council (ERC) under the European Union's
  Horizon 2020 research and innovation programme (grant agreement No 638992 -
  OPT4SMART)
  and from the 
  Italian grant SIR ``Scientific Independence of young Researchers'', project
  CoNetDomeSys, code RBSI14OF6H, funded by the Italian Ministry of Research and
  Education (MIUR).
  }
\thanks{ $^1$Ivano Notarnicola and Giuseppe Notarstefano are with the
  Department of Engineering, Universit\`a del Salento, Via Monteroni, 73100
  Lecce, Italy, \texttt{name.lastname@unisalento.it.}  } 
\thanks{$^2$Mauro Franceschelli (corresponding author) is with the Department of
  Electrical and Electronic Engineering, University of Cagliari, Piazza D'Armi,
  09123 Cagliari, Italy, \texttt{mauro.franceschelli@diee.unica.it.} }
% \thanks{$^*$Mauro Franceschelli is the corresponding author.} %
  }

\maketitle

\begin{abstract}
  In this paper we consider a distributed optimization scenario in which a set
  of processors aims at minimizing the maximum of a collection of ``separable
  convex functions'' subject to local constraints. This set-up is motivated by
  peak-demand minimization problems in smart grids. Here, the goal is to
  minimize the peak value over a finite horizon with: (i) the demand at each time
  instant being the sum of contributions from different devices, and (ii) the
  local states at different time instants being coupled through local
  dynamics. The min-max structure and the double coupling (through the devices
  and over the time horizon) makes this problem challenging in a distributed
  set-up (e.g., well-known distributed dual decomposition approaches cannot be
  applied). We propose a distributed algorithm based on the combination of
  duality methods and properties from min-max optimization. Specifically, we
  derive a series of equivalent problems by introducing ad-hoc slack variables
  and by going back and forth from primal and dual formulations. On the
  resulting problem we apply a dual subgradient method, which turns out to be a
  distributed algorithm. We prove the correctness of the proposed algorithm and
  show its effectiveness via numerical computations.
\end{abstract}

% In this paper we consider the problem of computing, with a distributed
% algorithm, the optimal solution to a convex optimization problem whose
% objective function is the maximum of the sum of a set of convex functions,
% each one representing a local objectives of agents in a multi-agent
% system. The considered objective function is not separable and standard
% distributed approaches can not be applied. The considered optimization problem
% generalizes the specific problem of coordinating a large set of electric
% thermal systems such as water heaters or freezers for electric demand side
% management purposes with an entirely distributed, private and cooperative
% approach which does not involve supervisors or information aggregators.

% \bigskip\bigskip In this paper we consider a constrained optimization problem
% in which we want to minimize the maximum of a set of convex functions subject
% to a set of coupling constraints.  Using a sequence of equivalent problem
% formulation obtained with suitable slack variables and duality tools, we are
% able to propose a fully distributed algorithm which solves the original
% min-max problem exhibiting the same behavior of the subgradient method.  We
% motivate the proposed scenario with a concrete problem set-up arising in the
% electric demand side management of electric thermal systems such as water
% heaters or freezers.  We show the effectiveness of the proposed strategy via
% numerical simulations.

\section{Introduction}
\label{sec:intro}
% A major challenge our society is facing is the need to reduce costs and
% pollution due to electric energy production. A significant part of these costs
% are related to the daily variation of the demand, since energy producers must
% employ peaking power plants which are generally less efficient than base-load
% power plants.
%
The addition of processing, measurement, communication and control
capability to the electric power grid is leading to smart grids, in which smart
generators, accumulators and loads can cooperate to execute Demand Side
Management (DSM) programs \cite{alizadeh2012demand}. The goal is to reduce the
hourly and daily variations and peaks of electric demand by optimizing
generation, storage and consumption.
%
% One way to address this challenge is to add flexibility to the hourly electric
% demand by executing so-called demand side management (DSM) programs, which aim
% at reducing the hourly and daily variations and peaks of electric demand by
% remotely controlling electric loads.
%
% These ideas are being pushed forward by the drive toward the so-called smart
% grid, i.e., the addition of modern processing, measurement, communication and
% control capability to the electric power grid.
%
% We point to the work by Alizadeh \emph{et al.} in \cite{alizadeh2012demand} for
% a description of how the modern tools developed by the scientific community in
% the context of information processing can be exploited for demand-side
% management in the smart grid.
%
A widely adopted objective in DSM programs is Peak-to-Average Ratio (PAR),
defined as the ratio between peak-daily and average-daily power demands. PAR
minimization gives raise to a min-max optimization problem if the average daily
electric load is assumed not to be affected by the demand response strategy.

% Recently there has been a surge of interest in distributed and decentralized
% solutions to DSM problems, i.e., a coordinator collects aggregate information
% about users' and electric devices needs and provides real-time prices or
% suggested behaviors while each user or smart-home autonomously schedule its own
% consumption based on the information provided by the coordinator.
% %
% In the research project CoNetDomeSys$^a$ it is being proposed to provide device
% level autonomous cooperation capability to achieve the objectives of DSM while
% not exploiting any form of centralized coordination, aggregate information on
% electric loads or real-time utility prices, some preliminary results are
% presented in~\cite{franceschelli2016coordination}.

% In the case of large populations of TCLs, the need to abide the internal
% temperature constraints does not allow to reduce appreciably the average power
% consumption but only to modulate it around an average value.
In \cite{mohsenian2010autonomous} the authors % Mohsenian-Rad \emph{et al.}
propose a game-theoretic model for PAR minimization and provide a distributed
energy-cost-based strategy for the users.
% which is proven to be optimal.
%
A noncooperative-game approach is also proposed in \cite{atzeni2013demand},
where optimal strategies are characterized and a distributed scheme is designed
based on a proximal decomposition algorithm.
%
%DSM
% In \cite{atzeni2013demand} the authors
% formulate a DSM problem in a smart grid as a noncooperative game and analyze the
% existence of optimal strategies. The authors present also a distributed
% algorithm to be run on the users' smart meters, which provides the optimal
% production and/or storage strategies, while preserving the privacy of the users
% and minimizing the required signaling with the central unit.
%
% The approach proposed our paper differs from approaches such as the one in
% \cite{atzeni2013demand} in that we are interested in cooperative solution to
% distributed optimization problems in the context of DSM, the optimal solution of
% a non-cooperative game is in general sub-optimal when cooperation is considered.
%
A key difference of the set-up in~\cite{mohsenian2010autonomous,atzeni2013demand}, 
compared to the one proposed in our paper, is that in those works each agent needs 
to know the total load and tariffs in the power distribution system. Moreover, the agents 
do not cooperate to compute the strategy.
In~\cite{parisio2014model} a Model Predictive Control scheme is proposed to
optimize micro-grid operations while satisfying a time-varying request and
operation constraints using a mixed-integer linear model.
%. The problem is modeled as a mixed-integer linear program.
% \MF{Other approaches for optimization-based management of smart grids involve
%   model predictive control. In~\cite{parisio2014model} the scenario of a network
%   of interconnected Microgrids which comprise generation capacities, storage
%   devices, and controllable loads, operating as a single controllable system
%   have been investigated. The authors propose an approach based on model
%   predictive control to the problem of efficiently optimizing microgrid
%   operations while satisfying a time-varying request and operation constraints
%   which can be formulated using mixed-integer linear programming
%   (MILP). Experimental results have been obtained as well.

In this paper we propose a novel distributed optimization framework for min-max
optimization problems commonly found in DSM problems.
Differently from the references above, we consider a cooperative, distributed
computation model in which the agents in the network do not have knowledge of
aggregate quantities, communicate only with neighboring agents and perform local
computations (with no central coordinator) to solve the optimization problem.
% despite the restrictive assumptions of limited locally exploitable information
% commonly adopted in the framework of multi-agent systems.

% STANDARD DUAL DECOMPOSITION FOR DISTRIBUTED OPTIMIZATION
Duality is a widely used tool for distributed optimization algorithms as shown,
e.g., in the tutorials~\cite{palomar2006tutorial,yang2010distributed}. 
% dual decomposition methods have been applied
% in order to develop distributed algorithms in a pure peer-to-peer set-up.
% In~\cite{yang2010distributed} and references therein a tutorial on network
% optimization via dual decomposition can be found.
%
% These standard approaches have been widely used in distributed optimization,
% but it
%
These standard approaches do not apply to the framework considered in this
paper.
%
% to optimization problems with coupling in
% the cost function and constraint set, as for the framework considered in this
% paper.
%
% problem set-up with coupling constraints
%\IN{In~\cite{chang2014distributed} distributed optimization methods are studied for
%solving locally constrained problems with a coupled global cost function.
%A consensus-based distributed primal-dual perturbation (PDP) algorithm is proposed
%in which agents use a consensus technique to estimate the global cost and
%constraint functions.}
%
%
In~\cite{chang2014distributed} a distributed consensus-based primal-dual
algorithm is proposed to solve optimization problems with coupled global cost
function and inequality constraints.

%
% Nedic Distributed constrained optimization by consensus-based primal-dual
% perturbation method \cite{chang2014distributed}
%
% Motivated by emerging applications in smart grid and distributed sparse
% regression, this paper studies distributed optimization methods for solving
% general problems which have a coupled global cost function and have inequality
% constraints. We consider a network scenario where each agent has no global
% knowledge and can access only its local mapping and constraint functions. To
% solve this problem in a distributed manner, we propose a consensus-based
% distributed primal-dual perturbation (PDP) algorithm. In the algorithm, agents
% employ the average consensus technique to estimate the global cost and
% constraint functions via exchanging messages with neighbors, and meanwhile use
% a local primal-dual perturbed subgradient method to approach a global optimum.

% MINIMAX
Min-max optimization is strictly related to saddle-point
problems. In~\cite{nedic2009subgradient} the authors propose a subgradient
method to generate approximate saddle-points.
A min-max problem is also considered in~\cite{srivastava2011distributed} and a
distributed algorithm based on a suitable penalty approach has been proposed. 
% The authors consider a time-varying communication graph and proved the almost
% sure convergence of the algorithm.
%Differently from our set-up, in \cite{srivastava2011distributed} each term of
%the max-function is local (and known by an agent).
%
Another class of algorithms exploits the exchange of active constraints among
the network nodes to solve constrained optimization problems which include
min-max problems,
\cite{notarstefano2011distributed,burger2014polyhedral}. Although they work
under asynchronous, directed communication they do not scale in set-ups as the
one in this paper in which the terms of the max function are coupled.
Very recently, in~\cite{mateos2015distributed} the authors proposed a
distributed projected subgradient method to solve constrained saddle-point
problems with agreement constraints. 
% The proposed algorithm is based on saddle-point dynamics with Laplacian averaging.
%
Although our problem set-up fits in those considered in~\cite{mateos2015distributed}, 
our algorithmic approach and the analysis are different.
In~\cite{simonetto2012regularized,koppel2015regret} saddle point dynamics are
used to design distributed algorithms for standard separable optimization
problems.

% CONTRIBUTION
% The main contributions of this paper are as follows. 
The contribution of this paper is twofold.  First, we propose a novel
distributed optimization framework which is strongly motivated by peak
power-demand minimization in DSM. The optimization problem has a min-max
structure with local constraints at each node. Each term in the max function
represents a daily cost (so that the maximum over a given horizon needs to be
minimized), while the local constraints are due to the local dynamics and input
bounds of the subsystems in the smart grid.
The problem is challenging when approached in a distributed way since it is
\emph{doubly coupled} (each term of the max function is coupled among the
agents, while the local constraints impose a coupling between different
``days'' in the time-horizon).

% Each term of the max function is coupled among the
% agents, since it is the sum of local functions each one known by the local agent
% only. Moreover, the local constraints impose a coupling between different
% ``days'' in the time-horizon.
%
%The goal is to solve the problem in a distributed way according to the computation model introduced above.
%
% Here each agent knows only its local constraint and its local objective function
% at each day.
% Also, the agents perform only local computations (no coordinator is
% present) and exchange information only with neighbors.

Second, as main paper contribution, we propose a distributed algorithm to solve
this class of min-max optimization problems. 
The algorithm has a very simple and clean structure in which a primal
minimization and a dual ascent step are performed. The primal problem has a
similar structure to the centralized one.
Despite this simple structure, which resembles standard distributed dual
methods, the algorithm is not a standard decomposition scheme and the derivation
of the algorithm is non-obvious.
Specifically, the algorithm is derived by heavily resorting to duality theory
and properties of min-max optimization (or saddle-point) problems. In
particular, a sequence of equivalent problems is derived in order to decompose
the originally coupled problem into locally-coupled subproblems, and
thus being able to design a distributed algorithm.
An interesting feature of the algorithm is its expression in terms of dual
variables of two different problems and of the original primal variables. Since
we apply duality more than once and on different problems, this property,
although apparently intuitive, was not obvious a priori.
%
%Another appealing feature is that each node only computes the decision variable
%of interest. Thus, problems which are both large-scale (many agents are
%present) and big-data (a large horizon is considered) can be solved.
%
Another appealing feature of the algorithm is that every limit point of the primal
sequence at each node is a (feasible) optimal solution of the original
optimization problem (although this is only convex and not strictly convex).
This property is obtained by the minimizing sequence of the local primal
subproblems without resorting to averaging schemes,
\cite{nedic2009approximate}. Finally, since each node only computes the decision
variable of interest, our algorithm can solve both large-scale (many agents are
present) and big-data (a large horizon is considered) problems.

% discuss differences with ... \cite{chang2014distributed}, \cite{nedic2009subgradient}
% \bigskip \bigskip \bigskip \bigskip \bigskip
% \paragraph*{Notation}
% We denote $\sx{i}\in\real^m$ the state of node $i$, while with $\sx{i}_\slotIndex$ denotes its $\slotIndex$-th component.
% We denote by $\1$ the vector $[1, \ldots, 1]^\top \in \real^m$.
%
%\IN{ADD: while intuitively it seems natural to construct the dual ... there are a lot of technical aspects 
%to be taken into account ... convergence and feasibility of solutions, regularities ...}
%
%\IN{
%The cost is not strictly convex... 
%}

The paper is structured as follows. In Section~\ref{sec:preliminaries} we
provide some useful preliminaries on optimization, duality theory and
subgradient methods. In Section~\ref{sec:distributed_algorithm} we formalize our
distributed min-max optimization set-up and present the main contribution of the
paper, a novel, duality based distributed optimization method. In
Section~\ref{sec:analysis} we characterize its convergence properties. Finally,
in Section \ref{sec:simulations} we corroborate the theoretical results with a
numerical example involving peak power minimization in a smart-grid scenario.

%Due to space constrains all proofs are omitted in this paper and 
%can be found in~\cite{notarnicola2016duality}.

Due to space constrains all proofs are omitted in this paper and 
will be provided in a forthcoming document.

\section{Preliminaries}
\label{sec:preliminaries}

\subsection{Optimization and Duality}
Consider a constrained optimization problem, addressed as primal problem,
having the form
\begin{align}
\begin{split}
  \min_{ z \in Z } \:& \: f(z)
  \\
  \subj \: & \: g(z) \preceq 0
%   \min_{ z \in Z } \: f(z) \hspace{0.8cm} \subj \: g(z) \preceq 0
\end{split}
\label{eq:preliminaries_primal}
\end{align}
where $Z \subseteq \real^N$ is a convex and compact set,
$\map{f}{\real^N}{\real}$ is a convex function and each component $\map{g_\slotIndex}{\real^N}{\real}$,
$\slotIndex \in \until{\slotUB}$, of $g$ is a convex function.

The following optimization problem
\begin{align}
\begin{split}
  \max_\mu \:& \: q(\mu)
  \\
  \subj \: & \: \mu \succeq 0
\end{split}
\label{eq:preliminaries_dual}
\end{align}
is called the dual of problem~\eqref{eq:preliminaries_primal}, where
$\map{q}{\real^\slotUB}{\real}$ is obtained by minimizing with respect to $z \in Z$
the Lagrangian function $\LL (z,\mu) := f(z) + \mu^\top g(z)$, i.e.,
$q(\mu) = \min_{z \in Z} \LL(z,\mu)$. Problem~\eqref{eq:preliminaries_dual} is
well posed since the domain of $q$ is convex and $q$ is concave on its domain.
% \begin{align}
%   \LL (x,\mu) := f(x) + \mu^\top g(x).
%   \label{eq:preliminaries_Lagrangian}
% \end{align}

It can be shown that the following inequality holds
\begin{align}
  \inf_{z \in Z } \sup_{\mu \succeq 0} \LL(z, \mu) \ge \sup_{\mu\succeq 0} \inf_{z \in X } \LL(z,\mu),
  \label{eq:preliminaries_weak_duality}
\end{align}
which is called weak duality.
When in~\eqref{eq:preliminaries_weak_duality} the equality holds, then we say
that strong duality holds and, thus, solving the primal
problem~\eqref{eq:preliminaries_primal} is equivalent to solving its dual
formulation~\eqref{eq:preliminaries_dual}. In this case the right-hand-side
problem in~\eqref{eq:preliminaries_weak_duality} is referred to as
\emph{saddle-point problem} of \eqref{eq:preliminaries_primal}.

\begin{definition}
  A pair $(z^\star , \mu^\star)$ is called an primal-dual optimal solution of
  problem~\eqref{eq:preliminaries_primal} if $z^\star\in Z$ and
  $\mu^\star\succeq 0$, and $(z^\star , \mu^\star)$ is a saddle point of the
  Lagrangian, i.e.,
\begin{align*}
  \LL (z^\star,\mu ) \le \LL (z^\star,\mu^\star) \le \LL (z,\mu^\star)
\end{align*}
for all $z\in Z$ and $\mu\succeq 0$.
\oprocend
\label{def:primal_dual_pair}
\end{definition}

% Note that for an optimal pair, $x^\star$ can be computed as $\argmin_{x\in X} \LL (x,\mu^\star)$
% and analogously $\mu^\star  = \argmax_{\mu \succeq 0} \LL(x^\star,\mu)$.

A more general min-max property can be stated. Let $Z \subseteq \real^N$ and
$W \subseteq \real^\slotUB$ be nonempty convex sets.  Let
$\phi : Z \times W \to \real$, then the following inequality
\begin{align*}
  \inf_{z\in Z} \sup_{w \in W}  \phi (z,w) \ge \sup_{w \in W} \inf_{z\in Z} \phi (z,w)
\end{align*}
holds true and is called the \emph{max-min} inequality. When the equality holds, then
we say that $\phi$, $Z$ and $W$ satisfy the \emph{strong max-min} property or the
\emph{saddle-point} property.

The following theorem gives a sufficient condition for the strong max-min property
to hold.

\begin{proposition}[{\cite[Propositions~4.3]{bertsekas2009min}}]
  Let $\phi$ be such that (i) $\phi (\cdot,w) : Z  \to \real$ is convex and closed
  for each $w \in W$, and (ii) $-\phi(z,\cdot) : W \to \real$ is convex and closed for
  each $z \in Z$.
  Assume further that $W$ and $Z$ are convex and compact sets. Then
  \begin{align*}
    \sup_{w \in W} \inf_{z \in Z} \phi (z,w) = \inf_{z \in Z} \sup_{w \in W}  \phi (z,w)
  \end{align*}
  and the set of saddle points is nonempty and compact.~\oprocend
  \label{prop:saddle_point}
\end{proposition}

%%%%%%%%%%%%%%%%%%%%%%%%%%%%%%%%%%%%%%%%%%%%%%%%%%%%%%%%%%%%%%%%%%%%%%%%%%%%%%%%

\subsection{Subgradient Method}
\label{sec:subgradient_method}
Consider the following (constrained) optimization problem
\begin{align}
  \min_{z\in Z} f (z)
  \label{eq:preliminaries_convex_problem}
\end{align}
with $Z \subseteq \real^N$ a closed convex set and $\map{f}{\real^N}{\real}$ convex. The
(projected) subgradient method is the iterative algorithm
% given by
\begin{align}
  z(t+1) = P_{Z} \Big( z(t) - \gamma(t) \widetilde{\nabla} f( z(t) ) \Big)
\label{eq:preliminaries_subgradient_iter}
\end{align}
where $t\in \mathbb{N}$ denotes the iteration index, $\gamma(t)$ is the
step-size, $\widetilde{\nabla} f( z(t) )$ denotes a subgradient of $f$ at
$z(t)$, and $P_Z(\cdot)$ is the Euclidean projection onto $Z$.

% \medskip
\smallskip

\begin{assumption}
\label{ass:step-size}
  The step-size $\gamma(t) \ge 0$ satisfies the following diminishing condition
  \begin{align}\notag
    \lim_{t\to \infty} \gamma(t) = 0, \:
    \sum_{t=1}^{\infty} \gamma(t) = \infty, \:
    \sum_{t=1}^{\infty} \gamma(t)^2 < \infty. \eqoprocend
  \end{align}
  % \oprocend
\end{assumption}

% \medskip
\smallskip

\begin{proposition}[{\cite[Proposition 3.2.6]{bertsekas2015convex}}]
  Assume that the subgradients $\widetilde{\nabla} f(z)$ are bounded for
  all $z \in Z$ and the set of optimal solutions is nonempty. Let the step-size
  $\gamma(t)\ge 0$ satisfy the diminishing condition in Assumption~\ref{ass:step-size}.
  % \begin{align}\notag
  %   \lim_{t\to \infty} \gamma(t) = 0, \:
  %   \sum_{t=1}^{\infty} \gamma(t) = \infty, \:
  %   \sum_{t=1}^{\infty} \gamma(t)^2 < \infty,
  % \end{align}
  Then the subgradient method in~\eqref{eq:preliminaries_subgradient_iter}
  applied to problem~\eqref{eq:preliminaries_convex_problem}
  converges in objective value and sequence $z(t)$ converges to an optimal
  solution.~\oprocend
  \label{prop:subgradient_convergence}
\end{proposition}

%%%%%%%%%%%%%%%%%%%%%%%%%%%%%%%%%%%%%%%%%%%%%%%%%%%%%%%%%%%%%%%%%%%%%%%%%%%%%%%%

\section{Problem Set-up and Distributed\\Optimization Algorithm}
\label{sec:distributed_algorithm}
In this section we set-up the distributed min-max optimization problem and
propose a distributed algorithm to solve it.

\subsection{Distributed min-max optimization set-up}
\label{sec:setup}
We consider a network of $N$ processors which communicate according to a
\emph{connected, undirected} graph $\GG = (\until{N}, \EE)$, where
$\EE\subseteq \until{N} \times \until{N}$ is the set of edges. That is, the edge
$(i,j)$ models the fact that node $i$ and $j$ exchange information.  We denote
by $\nbrs_i$ the set of \emph{neighbors} of node $i$ in the fixed graph $\GG$,
i.e., $\nbrs_i := \left\{j \in \until{N} \mid (i,j) \in \EE \right\}$.

Motivated by applications in Demand Side Management of Smart Grids, we introduce
a min-max optimization problem to be solved by the network processors in a
distributed way.
Specifically, we associate to each processor $i$ a decision
vector $\sx{i} = [ \sx{i}_1, \ldots, \sx{i}_\slotUB ]^\top \in \real^\slotUB$, a constraint
set $X_i\subseteq \real^\slotUB$ and local cost functions $g_{i\slotIndex}$,
$\slotIndex\in\until{\slotUB}$, and set-up the following optimization problem
\begin{align}
\begin{split}
  \min_{\sx{1}, \ldots, \sx{N}} \: & \: \max_{\slotIndex \in \until{\slotUB}} \sum_{i=1}^N g_{i \slotIndex}(\sx{i}_\slotIndex)
  \\
  \subj \: & \: \sx{i} \in X_i , \quad i\in\until{N}
\end{split}
\label{eq:minimax_starting_problem}
\end{align}
where for each $i\in\until{N}$ the set $X_i\subseteq\real^\slotUB$ is nonempty,
convex and compact, and the functions $g_{i \slotIndex} : \real \to \real $,
$\slotIndex\in\until{\slotUB}$, are convex.

Note that we use the superscript $i\in\until{N}$ to indicate that a vector
$\sx{i}\in\real^\slotUB$ belongs to node $i$, while we use the subscript to identify a
vector component, i.e., $\sx{i}_\slotIndex$, $\slotIndex\in\until{\slotUB}$, is the $\slotIndex$-th
component of $\sx{i}$.

% \begin{assumption}
%   The sets $X_i\subseteq\real^\slotUB$, $i\in\until{N}$, are nonempty convex sets.
%   \oprocend
% \label{ass:constraint_qualification}
% \end{assumption}

% \begin{assumption}
%   The functions $g_{i\slotIndex}$, $i\in\until{N}$ and $\slotIndex\in\until{\slotUB}$, are convex. \oprocend
% \label{ass:convexity}
% \end{assumption}

Using a standard approach for min-max problems, we introduce an auxiliary
variable $P$ to write the so called epigraph representation of
problem~\eqref{eq:minimax_starting_problem}, given by
\begin{align}
\begin{split}
  \min_{\sx{1}, \ldots, \sx{N}, P} \: & \: P
  \\
  \subj \: & \: \sx{i} \in X_i , \hspace{1.5cm} i\in\until{N}
  \\
  & \: \sum_{i=1}^N g_{i \slotIndex}(\sx{i}_\slotIndex) \le P, \hspace{0.3cm} \slotIndex \in \until{\slotUB}.
\end{split}
\label{eq:starting_problem}
\end{align}

% \GN{GN: ADD comment on problem structure in distributed set-up (not common).}

% We point out that optimization problem~\eqref{eq:starting_problem} has 
% a structure which is slightly different from the one commonly treated in distributed
% optimization. In fact, problem~\eqref{eq:starting_problem} has not a strictly convex
% and separable cost, therefore standard distributed dual decomposition approaches 
% cannot be applied as they are.
%
% Moreover, since the problem is (only) convex it has not a unique solution, and it is
% well known that this fact impacts on the dual approach, see e.g., \cite{nedic2009approximate} 
% % \cite{larsson1999ergodic}
% and reference therein. 

Notice that, this problem is convex, but not strictly convex. This means that it
is not guaranteed to have a unique solution. This impacts on dual approaches
when trying to recover a primal optimal solution, see e.g.,
\cite{nedic2009approximate}
% \cite{larsson1999ergodic}
and references therein.

% A standard dual decomposition scheme would require knowledge of the full
% network state, and standard distributed methods would require at least a
% shared memory among the agents to keep track of the updated values of the
% state vectors $\sx{i}$ for all $i\in \until{N}$.

%%%%%%%%%%%%%%%%%%%%%%%%%%%%%%%%%%%%%%%%%%%%%%%%%%%%%%%%%%%%%%%%%%%%%%%%%%%%%%

\subsection{Algorithm description}
\label{sec:alg_description}
Next, we introduce our distributed optimization algorithm.
Informally, the algorithm consists of a two-step procedure. First, each node
$i\in\until{N}$ stores a set of variables $((\sx{i}$, $\srho{i}), \smu{i})$
obtained as the primal-dual optimal solution pair of a local min-max optimization problem
with a structure similar to the centralized problem. The coupling with the other
nodes in the original formulation is replaced by a term depending on neighboring
variables $\slambda{ij}$, $j\in\nbrs_i$. These variables are updated in the
second step according to a suitable linear law weighting the difference of
neighboring $\smu{i}$.
Nodes use a diminishing step-size denoted by $\gamma(t)$ and can initialize the
variables $\slambda{ij}$, $j\in\nbrs_i$ to zero.
In the next table we formally state our \Dminmax/ distributed algorithm from the
perspective of node $i$.
\begin{algorithm}
\renewcommand{\thealgorithm}{}
\floatname{algorithm}{Distributed Algorithm}

  \begin{algorithmic}[0]
    \Statex \textbf{Processor states}: $(\sx{i}, \srho{i})$, $\smu{i}$ and $\slambda{ij}$ for $j\in\nbrs_i$

    \Statex \textbf{Evolution}:

      \StatexIndent[0.5] \textbf{Gather} $ \slambda{ji}(t)$ from $j\in\nbrs_i$

      \StatexIndent[0.5] \textbf{Compute} $\big((\sx{i}(t+1),\srho{i}(t+1)),\smu{i}(t+1)\big)$ as a primal-dual 
      optimal solution pair of
      \begin{align}
      \begin{split}
        \min_{\sx{i}, \srho{i}} \: & \: \srho{i}
        \\
        \subj \: & \: \sx{i} \in X_i
        \\
        & \: g_{i \slotIndex}(\sx{i}_\slotIndex ) +  \sum_{j\in\nbrs_i}  \big( \slambda{ij}(t) - \slambda{ji}(t) \big)_{\slotIndex} \le \srho{i},
        \\
        & \hspace{4.0cm}\slotIndex\in\until{\slotUB}
      \end{split}
      \label{eq:alg_minimization}
      \end{align}

      \StatexIndent[0.5] \textbf{Gather} $ \smu{j}(t+1)$ from $j\in\nbrs_i$

      \StatexIndent[0.5] \textbf{Update} for all $j\in\nbrs_i$
      \begin{align}
        \slambda{ij} (t \!+\! 1) = \slambda{ij}(t) - \gamma(t) ( \smu{i}(t \!+\! 1) \!-\! \smu{j}(t \!+\! 1))
      \label{eq:alg_update}
      \end{align}

  \end{algorithmic}
  \caption{\Dminmax/}
  \label{alg:distributed_DSM}
\end{algorithm}

The structure of the algorithm and the meaning of the updates will be clear in
the constructive analysis carried out in the next section.
At this point we want to point out that although
problem~\eqref{eq:alg_minimization} has the same min-max structure of
problem~\eqref{eq:starting_problem}, $\rho^i$ is not a copy of the centralized
cost $P$, but rather a local contribution to that cost. That is, as we will see,
the total cost $P$ will be the sum of the $\rho^i$s.

%%%%%%%%%%%%%%%%%%%%%%%%%%%%%%%%%%%%%%%%%%%%%%%%%%%%%%%%%%%%%%%%%%%%%%%%%%%%%%

\section{Algorithm Analysis}
\label{sec:analysis}
The analysis of the proposed \Dminmax/ distributed algorithm is 
constructive and heavily relies on duality theory tools.

We start by deriving the equivalent dual problem of~\eqref{eq:starting_problem}
which is formally stated in the next lemma.

\begin{lemma}
	The optimization problem
	\begin{align}
	\label{eq:dual_centr}
	\begin{split}
	  \max_{ \smu{} \in\real^\slotUB } \: & \: \sum_{i=1}^N q_i( \smu{} )
	  \\
	  \subj \: & \: \1^\top \smu{} = 1, \: \smu{} \succeq 0
	\end{split}
	\end{align}
	where $\1 := [1, \ldots, 1]^\top \in \real^\slotUB$ and
	\begin{align}
	  q_i( \smu{} ) & := \min_{\sx{i} \in X_i} \sum_{\slotIndex=1}^\slotUB \smu{}_\slotIndex g_{i \slotIndex}(\sx{i}_\slotIndex), 
	  \hspace{0.5cm} i\in\until{N},
	\label{eq:qi_definition}
	\end{align}
	is the dual of problem~\eqref{eq:starting_problem} and strong duality holds.\oprocend
\label{lem:equivalence_dual_and_initial}
\end{lemma}

In order to make problem~\eqref{eq:dual_centr} amenable for a distributed
solution, we can rewrite it in an equivalent form.
To this end, we introduce copies of the common optimization variable $\smu{}$ and
coherence constraints having the sparsity of the connected graph $\GG$, obtaining
\begin{align}
\begin{split}
  \max_{ \smu{1}, \ldots, \smu{N} }
  \: & \: \sum_{i=1}^N q_i( \smu{i} )
  \\
  \subj \: & \: \1^\top \smu{i} = 1,  \: \smu{i} \succeq 0, \: i\in\until{N}
  \\ & \:
  \smu{i} = \smu{j}, \: \text{for all } (i,j) \in \EE.
\end{split}
\label{eq:problem_with_copies}
\end{align}

Notice that we have also duplicated the simplex constraint so that it becomes a
local constraint for each node.

To solve this problem we can use a dual decomposition approach by designing a
dual subgradient algorithm. This can be done since the constraints are convex
and the cost function concave.
A dual subgradient algorithm applied to problem~\eqref{eq:problem_with_copies}
would immediately result into a distributed algorithm if functions
$q_i$ were available in a closed form.

Intuition suggests that deriving the dual of a dual problem would somehow bring
back to a primal formulation. However, we want to stress that:
\begin{enumerate}
\item problem \eqref{eq:problem_with_copies} is dualized rather than problem \eqref{eq:dual_centr},
\item different constraints are dualized, namely the coherence
constraints rather than the simplex ones.
\end{enumerate}

We start deriving the dual subgradient algorithm by dualizing only the coherence
constraints. Thus, we write the partial Lagrangian
\begin{align}
\begin{split}
  \LL_2( \smu{1},\ldots,  & \smu{N}, \{ \slambda{ij} \}_{(i,j) \in \EE })
  \\
  & = \sum_{i=1}^N \Big( q_i( \smu{i} ) + \sum_{j\in\nbrs_i} \slambda{ij}^\top (\smu{i} - \smu{j} ) \Big)
\end{split}
\label{eq:lagrangian_definition}
\end{align}
where $\slambda{ij} \in \real^\slotUB$ for all $(i,j)\in \EE$ are Lagrange multipliers associated to
the constraints $\smu{i} - \smu{j} = 0$.
By exploiting the undirected nature and the connectivity
of communication graph $\GG$, after some algebraic manipulations, we get
\begin{align}
\begin{split}
  \LL_2( \smu{1}, \ldots, & \smu{N},  \{ \slambda{ij} \}_{(i,j) \in \EE })
  \\
  & = \sum_{i=1}^N \Big( q_i( \smu{i} ) + {\smu{i} }^\top \!\! \sum_{j\in\nbrs_i} (\slambda{ij} - \slambda{ji})  \Big),
\end{split}
\label{eq:lagrangian_rearrangement}
\end{align}
which is separable with respect to $\smu{i}$, $i\in\until{N}$.

The dual of problem~\eqref{eq:problem_with_copies} is thus
\begin{equation}
  \min_{ \{\slambda{ij}\}_{(i,j)\in\EE} } \eta(\{\slambda{ij}\}_{(i,j)\in\EE}) = \sum_{i=1}^N
  \eta_i\left(\{\slambda{ij},\slambda{ji}\}_{j\in\nbrs_i}\right),
\label{eq:dual_dual}
\end{equation}
where for all $i\in\until{N}$
\begin{align*}
  \eta_i (\{\slambda{ij},\slambda{ji}\}_{j\in\nbrs_i} )
  \! = \! \max_{ \1^\top \smu{i} = 1,\smu{i} \succeq 0 } \!  q_i( \smu{i} ) \! 
  +\! {\smu{i} }^\top\!\!\! \sum_{j\in\nbrs_i} (\slambda{ij} \! -\!  \slambda{ji}).
\end{align*}

In order to apply a subgradient method to problem~\eqref{eq:dual_dual}, we
recall, \cite[Section~6.1]{bertsekas1999nonlinear}, that
\begin{align}
\frac{\tilde \partial
   \eta ( \{\slambda{ij}\}_{(i,j)\in\EE} ) }{ \partial \slambda{ij} }
  = {\smu{i}}^\star - {\smu{j}}^\star,
\label{eq:eta_subgradient}
\end{align}
where $\frac{\tilde \partial \eta (\cdot)}{\partial \slambda{ij}}$ denotes the component
associated to the variable $\slambda{ij}$ of a subgradient of $\eta$, and
\begin{align*}
  {\smu{k}}^\star \in \argmax_{ \1^\top \smu{k} = 1,\smu{k} \succeq 0 } \bigg( q_k( \smu{k} ) + 
  {\smu{k}}^\top \sum_{h\in\nbrs_k} (\slambda{kh} - \slambda{hk}) \bigg),
\end{align*}
for $k=i,j$.
%
%Also, for the method to converge a diminishing step-size, satisfying
%Assumption~\ref{ass:step-size}, is needed.
%
The dual subgradient algorithm for problem~\eqref{eq:problem_with_copies} can be
summarized as follows, for each node $i\in\until{N}$:
\begin{itemize}
\item[(S1)]\label{item:s1} 
  receive $\slambda{ji}(t)$, $j \in \nbrs_i$ and
  compute a subgradient $\smu{i}(t+1)$ by solving
\begin{align}
  \begin{split}
    \max_{\smu{i} } \: & \: q_i( \smu{i} ) +
    {\smu{i}}^\top \!\! \sum_{j\in\nbrs_i} (\slambda{ij}(t) - \slambda{ji}(t))
    \\
    \subj \:   & \: \1^\top \smu{i} = 1,\smu{i} \succeq 0.
  \end{split}
\label{eq:dual_subgradient}
\end{align}
\item[(S2)]\label{item:s2}
  exchange with neighbors the updated $\smu{j}(t+1)$, $j \in \nbrs_i$, 
  and update $\slambda{ij}$, $j\in\nbrs_i$, via
  \begin{align*}
    \slambda{ij} (t \!+\! 1) = \slambda{ij}(t) - \gamma(t) (\smu{i}(t \!+\! 1) \!-\! \smu{j}(t + 1)).
  \end{align*}
  where $\gamma (t)$ denotes the step-size.
\end{itemize}

It is worth noting that in~\eqref{eq:dual_subgradient} the value of
$\slambda{ij}(t)$ and $\slambda{ji}(t)$, for $j\in\nbrs_i$, is fixed as
highlighted by the index $t$.  
Moreover, we want to stress, once again, that the algorithm is \emph{not}
implementable as it is written, since functions $q_i$ are not available in
closed form.
On this regard, here we slightly abuse notation since in (S1)-(S2) we use
$\smu{i}(t)$ as in the \Dminmax/ algorithm, but we have not proven
the equivalence yet. Since we will prove it in the next lemmas we preferred
not to overweight the notation.

% The next result establishes the \emph{formal} correctness of the proposed strategy.
\begin{lemma}
  The dual subgradient updates (S1)-(S2), with step-size $\gamma(t)$ 
  satisfying Assumption~\ref{ass:step-size}, generate
  sequences $\{ \slambda{ij}(t) \}$, $(i,j)\in \EE$ that converge in objective value to 
  $\eta^\star = q^\star = P^\star$, optimal costs of~\eqref{eq:dual_dual},~\eqref{eq:dual_centr} 
  and~\eqref{eq:starting_problem}, respectively.\oprocend
%  the optimal cost $\eta^\star$ of problem~\eqref{eq:dual_dual}.\oprocend
  \label{lem:dual_subgradient_correcteness}
\end{lemma}

We can explicitly rephrase update~\eqref{eq:dual_subgradient} by plugging in the
definition of $q_i$, given in~\eqref{eq:qi_definition}, thus obtaining the
following max-min optimization problem
\begin{align}
  \hspace{-0.2cm}
  \max_{ \1^\top \smu{i} = 1,\smu{i} \succeq 0 } \! \bigg(\!
  \min_{\sx{i} \in X_i}
  \sum_{\slotIndex=1}^\slotUB \smu{i}_\slotIndex \Big(
  g_{i \slotIndex}( \sx{i}_\slotIndex) \! + \!\! \sum_{j\in\nbrs_i} \!\! (\slambda{ij}(t) \!-\! \slambda{ji}(t))_\slotIndex \! \Big) \!\! \bigg).
\label{eq:maxmin}
\end{align}
Notice that this is a local problem at each node $i$ once the value for
$\slambda{ij}(t)$ and $\slambda{ji}(t)$ for all $j\in\nbrs_i$ are given.

%  Equivalence between local primal problem and minmax
\begin{lemma}
  Max-min optimization problem~\eqref{eq:maxmin} is the saddle point problem
  associated to problem~\eqref{eq:alg_minimization}. Moreover, a primal-dual
  optimal solution pair of~\eqref{eq:alg_minimization}, call it $\{ ( \sx{i}(t+1) , \srho{i}(t+1) ), \smu{i}(t+1) \}$, 
  exists and $(\sx{i}(t+1),\smu{i}(t+1))$ is a solution of~\eqref{eq:maxmin}.
\label{lem:dual_minmax_equivalence}
\end{lemma}%
\begin{proof}
We give a constructive proof which clarifies how problem~\eqref{eq:alg_minimization}
is derived from~\eqref{eq:maxmin}.
Define
\begin{align}
  \phi(\sx{i},\smu{i}):=\sum_{\slotIndex=1}^\slotUB \smu{i}_\slotIndex \Big(
  g_{i \slotIndex}( \sx{i}_\slotIndex) \! + \!\! \sum_{j\in\nbrs_i} (\slambda{ij}(t) \!-\! \slambda{ji} (t) )_\slotIndex \Big)
\end{align}
and note that (i) $\phi(\cdot,\smu{i})$ is closed and convex for all
$\smu{i} \succeq 0$ and (ii) $\phi(\sx{i}, \cdot )$ is closed and concave (linear over the compact 
$\1^\top \smu{i} = 1$, $\smu{i}\succeq 0$), for
all $\sx{i}\in\real^\slotUB$.
Thus we can invoke the saddle point Proposition~\ref{prop:saddle_point} which allows us to switch the max 
and min operators, and write
\begin{align}
\begin{split}
  & \max_{ \1^\top \smu{i} = 1,\smu{i} \succeq 0 } \! \bigg(\!
  \min_{\sx{i} \in X_i}
  \sum_{\slotIndex=1}^\slotUB \smu{i}_\slotIndex \Big(
  g_{i \slotIndex}( \sx{i}_\slotIndex) \! + \!\!\! \sum_{j\in\nbrs_i} \!\! (\slambda{ij}(t) \!-\! \slambda{ji}(t) )_\slotIndex \! \Big) \!\! \! \bigg)
  \\
  & \!\!=\!\!\!
  \min_{\sx{i} \in X_i} \!\! \bigg(\!
  \max_{ \1^\top \smu{i} = 1, \smu{i} \succeq 0 }
  \sum_{\slotIndex=1}^\slotUB \smu{i}_\slotIndex \Big(
  g_{i \slotIndex}(\sx{i}_\slotIndex) \! + \!\!\! \sum_{j\in\nbrs_i} \!\! (\slambda{ij}(t) \!-\!  \slambda{ji}(t) )_\slotIndex \! \Big) \!\!\! \bigg)\!.
\end{split}
  \label{eq:minmax}
\end{align}

Since the inner maximization problem depends nonlinearly on $\sx{i}$
(which is itself an optimization variable), the solution cannot be obtained without
considering the optimization also on $\sx{i}$. We overcome this issue
% by rephrasing the maximization via its dual as follows.
by substituting the inner maximization problem with its equivalent dual.
Notice that the inner problem is a linear program when $\sx{i}$ are kept
fixed, and thus strong duality can be exploited.
Introducing a scalar multiplier $\srho{i}$ associated to the simplex constraint,
we have %that
\begin{align}
\begin{split}
  \max_{\smu{i} } \: & \:
     \sum_{\slotIndex=1}^\slotUB \smu{i}_\slotIndex \Big(
        g_{i \slotIndex}(\sx{i}_\slotIndex) + \sum_{j\in\nbrs_i} (\slambda{ij}(t) - \slambda{ji}(t) )_{\slotIndex}
      \Big)
  \\
  \subj \: & \: \1^\top \smu{i} = 1, \smu{i} \succeq 0
\end{split}
\label{eq:intermediate_problem}
\end{align}
is equivalent to its dual
\begin{align}
\label{eq:rho_formulation}
  \min_{\srho{i}} \: & \: \srho{i}
  \\
\notag
  \subj \: 
   & \: g_{i \slotIndex}(\sx{i}_\slotIndex ) \!+\!\! \sum_{j\in\nbrs_i} 
   \! (\slambda{ij}(t) \!-\! \slambda{ji}(t) )_{\slotIndex} \le \srho{i}, \slotIndex \!\in \! \until{\slotUB}
\end{align}
where the $\slotUB$ inequality constraints follow from the minimization of the
partial Lagrangian of~\eqref{eq:intermediate_problem} with respect to
$\smu{i} \succeq 0$.
Plugging formulation~\eqref{eq:rho_formulation} in place of the inner maximization
in~\eqref{eq:minmax}, we can write a \emph{joint} minimization, i.e.,
minimize simultaneously with respect to $\sx{i}$ and $\srho{i}$, which leads to~\eqref{eq:alg_minimization}. 

To prove the second part, notice that problem~\eqref{eq:alg_minimization} is convex.
Then, the problem satisfies the Slater's constraint qualification and, thus, strong duality
holds. Therefore, a primal-dual optimal solution pair
$(\sx{i}(t+1), \srho{i}(t+1),\smu{i}(t+1))$ exists and from the previous
arguments the proof follows.%
\end{proof}

We point out that the previous lemma shows that performing 
minimization in~\eqref{eq:alg_minimization} turns out to be equivalent to 
performing step~(S1).

We are now ready to state the main result of the paper, namely the convergence 
of the \Dminmax/ distributed algorithm.

\begin{theorem}
  Let $\{ (\sx{i}(t), \srho{i}(t))\}$, $i\in\until{N}$, be
  the sequence generated by the \Dminmax/ distributed algorithm,
  with $\gamma(t)$ satisfying Assumption~\ref{ass:step-size}.
  Then, the sequence $\{\sum_{i=1}^N \rho^i(t)\}$ converges to the optimal cost $P^\star$
  of~\eqref{eq:minimax_starting_problem} and
  every limit point of the sequence $\{ \sx{i}(t) \}$, $i\in\until{N}$, is an 
  optimal (feasible) solution of~\eqref{eq:minimax_starting_problem}.
\oprocend
\label{thm:convergence}
\end{theorem}

\begin{remark}
  From condition~\eqref{eq:minmax} it can be shown that
  each $\srho{i}(t)$ is equal to $\eta_i( \{\slambda{ij},\slambda{ji}\}_{j\in\nbrs_i} )$ for all $ t \ge 0$. Since the optimal cost of \eqref{eq:dual_dual} is equal
  to the optimal primal cost $P^\star$, then we have that
  $\lim_{t\to\infty}\sum_i \rho^i(t) = \lim_{t\to\infty}\sum_i \eta_i(t) =
  P^\star$.  \oprocend
\label{rem:sum_rho_P}
\end{remark}

%%%%%%%%%%%%%%%%%%%%%%%%%%%%%%%%%%%%%%%%%%%%%%%%%%%%%%%%%%%%%%%%%%%%%%

\section{Numerical Simulations}
\label{sec:simulations}
In this section we propose a numerical example in which we apply the proposed
method to a network of Thermostatically Controlled Loads (TCLs) (such as air
conditioners, heat pumps, electric water heaters), \cite{Alizadeh2015reduced}.
% which minimize the peak electric energy consumption of a set of electric water heaters. 
% There is an extensive literature on the modeling of domestic electric water
% heaters (DEWH) (\cite{paull2010novel}) and, in general, Thermostatically
% Controlled Loads (TCLs), \cite{Alizadeh2015reduced}.

The dynamical model of the $i$-th device is given by 
% consists of a single state variable,
% the temperature $0\leq T_i(\tau)\in \real$, which evolves according to the dynamics
\begin{equation}
  \label{eq:agent_model}
  \dot{T}^{i}(\tau)= -\alpha \left(T^{i}(\tau)-T^{i}_{out}(\tau)\right)+ Q\sx{i}(\tau),
\end{equation}
where $T_i(\tau)\geq0$ is the temperature, $\alpha > 0$ is a parameter depending
on geometric and thermal characteristics, $T^{i}_{out}(\tau)$ is the air
temperature outside the device, $\sx{i}(\tau)\in \left[0,1\right]$ is the
control input, and $Q>0$ is a scaling factor.
%  which
% may vary in the interval $\left[0,1\right]$ that represents the throttle level
% of the heater, $0$ when switched off and $1$ at full power.
% In this example we simplify the problem by considering devices that can
% be regulated to an arbitrary set point in the interval.
%
%In most electric water heaters the heating element can be set only to on and off states, in this paper we consider the case in which the power injected by the heating element can be modulated by pulse-width modulation of its on-off states to achieve the desired effective power level or the input current in the heating element can be regulated to an arbitrary set point in the interval.
%
%Thus, given the water temperature $T(\tau)$ at time $\tau$, if we consider an interval of time in which the heating is constant ($x(t)=x$), a solution to system \eqref{eq:agent_model} is
%%
%\begin{equation}
%  \label{eq:solution}
%  T(t)= T(\tau)e^{-\alpha (t-\tau)}+\left(1-e^{-\alpha (t-\tau)}\right)\left(-T_{out} +\frac{k}{\alpha} x\right).
%\end{equation}

% A first constraint for the control input is to keep the temperature $T^i(\tau)$
% within a given range $T^i(\tau)\in \left[T_{min},T_{max}\right]$.

% A simple way to manipulate this model in a suitable form to numerically optimize the control
% input $\sx{i}(t)$ is to

We consider a discretized version of the system with constant input over the
sampling interval $\Delta\tau$, i.e., $x^{i}(\tau)=\sx{i}_{\slotIndex}$ for
$\tau\in \left[\slotIndex \Delta\tau, (\slotIndex+1) \Delta\tau \right)$, and
sampled state $T^{i}_{\slotIndex}$,
%
% assume the inputs to be constant over the sampling interval $\Delta\tau$,
% i.e., $x^{i}(\tau)=\sx{i}_{\slotIndex}$ for $\tau\in \left[\slotIndex \Delta\tau, (\slotIndex+1) \Delta\tau \right)$
% with $\slotIndex =1,\ldots,\slotUB$. The discrete-time version of system
% \eqref{eq:agent_model} is
%
\begin{equation}
  T^{i}_{\slotIndex+1}= T^{i}_{\slotIndex}e^{-\alpha \Delta\tau}+\left(1-e^{-\alpha \Delta\tau}\right)\left(\frac{Q}{\alpha} \sx{i}_{\slotIndex}-T^{i}_{out,\slotIndex}\right).
\label{eq:agent_model_discrete_time}
\end{equation}
%

% The aggregate power consumption of the set of $N$ devices in the generic time slot $[ \slotIndex dt, (\slotIndex+1) dt ]$ is
% $\sum_{i=1}^N g_{i \slotIndex} ( \sx{i}_\slotIndex )$,
% %
% %\begin{equation*}
% %  f_\slotIndex(x) = \sum_{i=1}^N g_{i \slotIndex} \left(\sx{i}_\slotIndex\right),
% %\end{equation*}
% where $g_{i \slotIndex} ( \cdot )$ represents the power consumption of the device for the given throttle level
% $\sx{i}_\slotIndex$ \IN{and is not known by any agent}. 

We assume that the power consumption $g_{i \slotIndex} ( \sx{i}_\slotIndex)$ of the $i$-th
device in the $\slotIndex$-th slot $[ \slotIndex \Delta\tau, (\slotIndex+1) \Delta\tau ]$ is directly proportional
to $\sx{i}_\slotIndex$.  For the sake of simplicity we consider
$g_{i \slotIndex} (\sx{i}_\slotIndex ) = \sx{i}_\slotIndex$ in the numerical example proposed in
this section.
Thus, optimization problem \eqref{eq:minimax_starting_problem} for this scenario is
% objective of the cooperation among the set
% of devices is to minimize their peak power consumption over a finite
% time horizon, i.e.,
%
\begin{align} \label{eq:simulated_problem}
\begin{split}
  \min_{\sx{1}, \ldots, \sx{N}} \: & \: \max_{\slotIndex \in \until{\slotUB}} \sum_{i=1}^N \sx{i}_\slotIndex
  \\
  \subj \: & \: \sx{i} \in X_i , \quad i\in\until{N}
\end{split}
\end{align}
where
$X_i := \{\sx{i} \in \real^\slotUB \mid A_i \sx{i} \preceq b_i \text{ and }
\sx{i}\in[0,1]^\slotUB \}$,
with $A_i$ and $b_i$ obtained by enforcing the dynamics
constraints \eqref{eq:agent_model_discrete_time} and temperature constraints
$T^i_\slotIndex \in \left[T_{min}, T_{max}\right]$.

%To construct $A_i$ and $b_i$, let us denote $\hat{A}=e^{-\alpha \Delta\tau}$ and
%$\hat{B}=1-e^{-\alpha \Delta\tau}$. We can compute the trajectory of $T^{i}_\slotIndex$ as
%a function of $\sx{i}_{\slotIndex}$, $T^{i}_{out,\slotIndex}$ and $T^{i}_0$ as follows. Let
%$\bar{T}^{i},\bar{T}^{i}_{out}\in \real^\slotUB$ be vectors whose $\slotIndex$-th element
%corresponds respectively to $T^{i}_\slotIndex$ and $T^{i}_{out,\slotIndex}$. Then, based
%on~\eqref{eq:agent_model_discrete_time} it holds
%%
%\begin{align*}
%\setlength{\arraycolsep}{3pt}
%  \bar{T}^i \! =\!\!
%  \underbrace{
%  \begin{bmatrix}
%    \hat{B} &  0 & \ldots & 0 \\
%    \hat{A}\hat{B} & \hat{B} & \ldots & 0 \\
%    \vdots \\
%    \hat{A}^{\slotUB}\hat{B}  & \hat{A}^{\slotUB-1}\hat{B}  & \ldots & \hat{B}
%  \end{bmatrix}
%  }_{F} \!\!
%  \Big(\!\! -\bar{T}^i_{out}+\frac{Q}{\alpha} \sx{i} \Big)
%  \!+\!
%  \underbrace{
%  \begin{bmatrix}
%    \hat{A} \\
%    \hat{A}^2 \\
%    \cdots \\
%    \hat{A}^\slotUB  \\
%  \end{bmatrix}
%  }_{G} \! T^{i}_0.
%\end{align*}
%
%Thus, the matrix $A_i$ and the vector $b_i$ turn out to be
%\begin{align*}
%  A_i =
%  \begin{bmatrix}
%    \frac{k}{\alpha}F\\
%    -\frac{k}{\alpha} F
%  \end{bmatrix},
%  \quad
%  b_i =
%  \begin{bmatrix}
%    T_{max} \1 -G T^i_0 + F \bar{T}^i_{out} \\
%    -T_{min} \1 +G T^i_0 -F \bar{T}^i_{out}
%  \end{bmatrix}.
%\end{align*}
%%
%% Moreover, the control input $x^i(\slotIndex)$ in our model is limited to the interval 
%% $[0,1]$, therefore also the constraint $x^i\in  [0,1]^\slotUB$ needs to be taken into account.
%%

In the proposed numerical example we consider $N=15$ agents communicating
according to an undirected connected Erd\H{o}s-R\'enyi random graph $\GG$ with
parameter $0.2$. We consider a horizon of $\slotUB=50$.  Finally, a diminishing
step-size sequence $\gamma(t) = (\frac{1}{t})^{0.8}$ at iteration $t$, which
satisfies Assumption~\ref{ass:step-size}, is used.

% CONVERGENCE
In Figure~\ref{fig:rho} we show the evolution at each algorithm iteration $t$ of
the local objective functions $\srho{i}(t)$, $i\in\until{N}$, (solid lines)
which converge to stationary values.  We also plot their sum
$\sum_{i=1}^N\srho{i}(t)$ (dotted line) that asymptotically converges to the
centralized optimal cost $P^\star$ of problem~\eqref{eq:simulated_problem}
(see Remark~\ref{rem:sum_rho_P}).

\begin{figure}[!htbp]
\centering
  \includegraphics[scale=0.85]{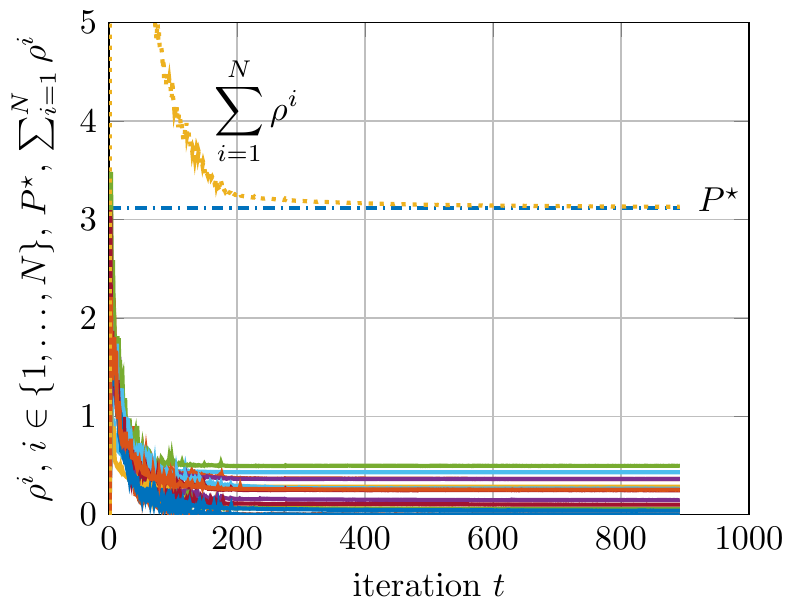}
  \caption{
    Evolution of $\srho{i}$, $i\in\until{N}$, (solid lines), their sum
    $\sum_{i=1}^N\srho{i}$ (dotted line), and (centralized) optimal cost
    $P^\star$ (dash-dotted line).
    }
  \label{fig:rho}
\end{figure}

% BEHAVIOR
In Figure~\ref{fig:xx} it is shown the profile of an optimal consumption of the devices,
i.e., $\sum_{i=1}^N {\sx{i}_\slotIndex}^\star$, over the
horizon $\slotIndex=1,\ldots,\slotUB$. It can be seen that the proposed method effectively
shaves off the peak power demand.
In the same figure it also shown an optimal consumption strategy ${\sx{i}}^\star$ that each single device locally computes.
%it can be seen that the network of water heaters have de-synchronized power demand compatibly with the constraint of keeping their
%internal temperature within the given interval $\left[T_{min}, T_{max}\right]$ despite a time-varying external temperature.
\begin{figure}[!htbp]
\centering
  \includegraphics[scale=0.85]{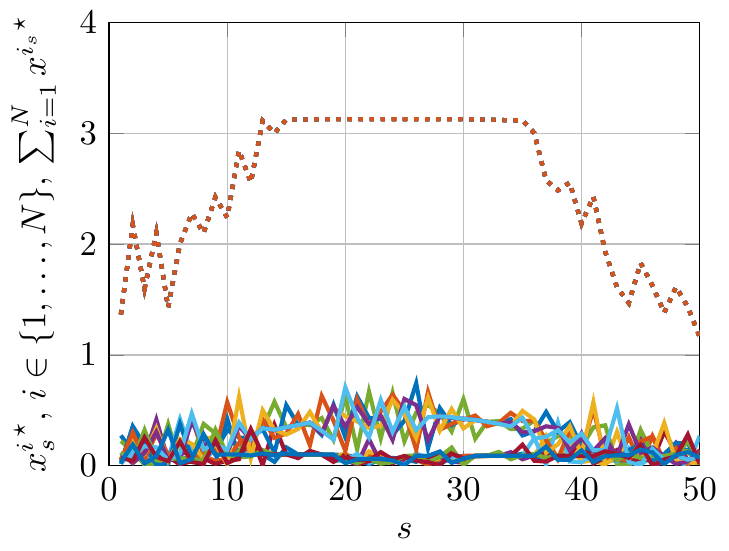}
  \caption{
    Profile of optimal solutions ${\sx{i}}^\star$ (solid lines), and $\sum_{i=1}^N{\sx{i}_\slotIndex}^\star$ 
    (dotted line) on the optimization horizon $\until{\slotUB}$.
    }
  \label{fig:xx}
\end{figure}

% RATE OF CONVERGENCE
Finally, in Figure~\ref{fig:cost} it is shown the convergence rate of the
distributed algorithm, i.e., the difference between the centralized optimal
cost  $P^\star$ and the sum of the local costs $\sum_{i=1}^N \srho{i} (t)$,
in logarithmic scale.  It can be seen that the proposed algorithm converges to
the optimal cost with sublinear rate $O(1/\sqrt{t})$ as expected for a
subgradient method.
\begin{figure}[!htbp]
\centering
  \includegraphics[scale=0.85]{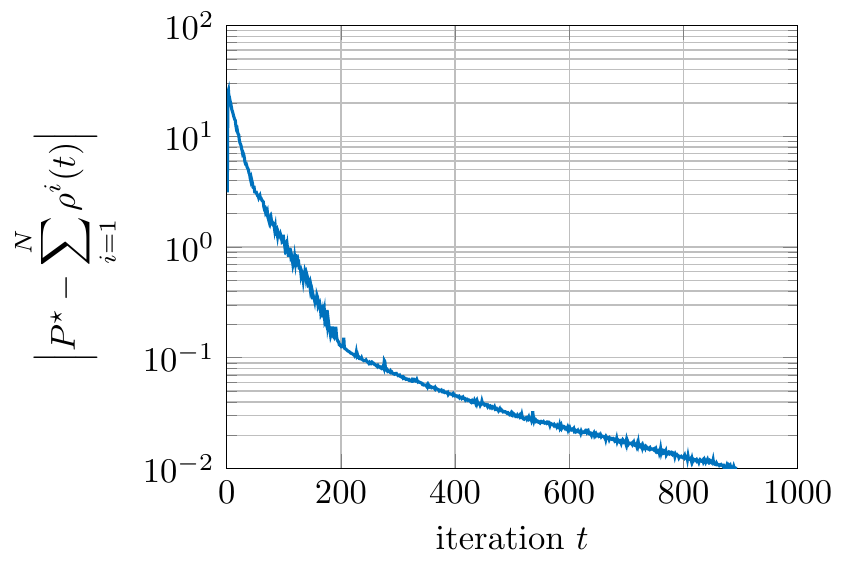}
  \caption{
    Evolution of the cost error, in logarithmic scale.
    }
  \label{fig:cost}
\end{figure}

%%%%%%%%%%%%%%%%%%%%%%%%%%%%%%%%%%%%%%%%%%%%%%%%%%%%%%%%%%%%%%%%%%%%%%

\section{Conclusions} %
\label{sec:conclusions} %
In this paper we have introduced a novel distributed min-max optimization
framework motivated by peak minimization problems in Demand Side
Management. Standard distributed optimization algorithms cannot be applied to
this problem set-up due to a highly nontrivial coupling in the objective
function and in the constraints.
We proposed a distributed algorithm based on the combination of duality methods
and properties from min-max optimization.
We proved the correctness of the proposed algorithm and corroborated the
theoretical results with a numerical example.

\begin{small}
  \bibliographystyle{IEEEtran}
  \bibliography{distributed_min-max}

\begin{thebibliography}{10}
\providecommand{\url}[1]{#1}
\csname url@samestyle\endcsname
\providecommand{\newblock}{\relax}
\providecommand{\bibinfo}[2]{#2}
\providecommand{\BIBentrySTDinterwordspacing}{\spaceskip=0pt\relax}
\providecommand{\BIBentryALTinterwordstretchfactor}{4}
\providecommand{\BIBentryALTinterwordspacing}{\spaceskip=\fontdimen2\font plus
\BIBentryALTinterwordstretchfactor\fontdimen3\font minus
  \fontdimen4\font\relax}
\providecommand{\BIBforeignlanguage}[2]{{%
\expandafter\ifx\csname l@#1\endcsname\relax
\typeout{** WARNING: IEEEtran.bst: No hyphenation pattern has been}%
\typeout{** loaded for the language `#1'. Using the pattern for}%
\typeout{** the default language instead.}%
\else
\language=\csname l@#1\endcsname
\fi
#2}}
\providecommand{\BIBdecl}{\relax}
\BIBdecl

\bibitem{alizadeh2012demand}
M.~Alizadeh, X.~Li, Z.~Wang, A.~Scaglione, and R.~Melton, ``Demand-side
  management in the smart grid: Information processing for the power switch,''
  \emph{IEEE Signal Processing Magazine}, vol.~29, no.~5, pp. 55--67, 2012.

\bibitem{mohsenian2010autonomous}
A.-H. Mohsenian-Rad, V.~W. Wong, J.~Jatskevich, R.~Schober, and A.~Leon-Garcia,
  ``Autonomous demand-side management based on game-theoretic energy
  consumption scheduling for the future smart grid,'' \emph{IEEE Transactions
  on Smart Grid}, vol.~1, no.~3, pp. 320--331, 2010.

\bibitem{atzeni2013demand}
I.~Atzeni, L.~G. Ord{\'o}{\~n}ez, G.~Scutari, D.~P. Palomar, and J.~R.
  Fonollosa, ``Demand-side management via distributed energy generation and
  storage optimization,'' \emph{IEEE Transactions on Smart Grid}, vol.~4,
  no.~2, pp. 866--876, 2013.

\bibitem{parisio2014model}
A.~Parisio, E.~Rikos, and L.~Glielmo, ``A model predictive control approach to
  microgrid operation optimization,'' \emph{IEEE Transactions on Control
  Systems Technology}, vol.~22, no.~5, pp. 1813--1827, 2014.

\bibitem{palomar2006tutorial}
D.~P. Palomar and M.~Chiang, ``A tutorial on decomposition methods for network
  utility maximization,'' \emph{IEEE Journal on Selected Areas in
  Communications}, vol.~24, no.~8, pp. 1439--1451, 2006.

\bibitem{yang2010distributed}
B.~Yang and M.~Johansson, ``Distributed optimization and games: A tutorial
  overview,'' in \emph{Networked Control Systems}.\hskip 1em plus 0.5em minus
  0.4em\relax Springer, 2010, pp. 109--148.

\bibitem{chang2014distributed}
T.-H. Chang, A.~Nedi{\'c}, and A.~Scaglione, ``Distributed constrained
  optimization by consensus-based primal-dual perturbation method,'' \emph{IEEE
  Transactions on Automatic Control}, vol.~59, no.~6, pp. 1524--1538, 2014.

\bibitem{nedic2009subgradient}
A.~Nedi{\'c} and A.~Ozdaglar, ``Subgradient methods for saddle-point
  problems,'' \emph{Journal of optimization theory and applications}, vol. 142,
  no.~1, pp. 205--228, 2009.

\bibitem{srivastava2011distributed}
K.~Srivastava, A.~Nedi{\'c}, and D.~Stipanovi{\'c}, ``Distributed min-max
  optimization in networks,'' in \emph{IEEE 17th International Conference on
  Digital Signal Processing (DSP)}, 2011, pp. 1--8.

\bibitem{notarstefano2011distributed}
G.~Notarstefano and F.~Bullo, ``Distributed abstract optimization via
  constraints consensus: Theory and applications,'' \emph{IEEE Transactions on
  Automatic Control}, vol.~56, no.~10, pp. 2247--2261, October 2011.

\bibitem{burger2014polyhedral}
M.~B{\"u}rger, G.~Notarstefano, and F.~Allg{\"o}wer, ``A polyhedral
  approximation framework for convex and robust distributed optimization,''
  \emph{IEEE Transactions on Automatic Control}, vol.~59, no.~2, pp. 384--395,
  2014.

\bibitem{mateos2015distributed}
D.~Mateos-N{\'u}{\~n}ez and J.~Cort{\'e}s, ``Distributed subgradient methods
  for saddle-point problems,'' in \emph{IEEE 54th Conference on Decision and
  Control (CDC)}, 2015.

\bibitem{simonetto2012regularized}
A.~Simonetto, T.~Keviczky, and M.~Johansson, ``A regularized saddle-point
  algorithm for networked optimization with resource allocation constraints,''
  in \emph{IEEE 51st Conference on Decision and Control (CDC)}, 2012, pp.
  7476--7481.

\bibitem{koppel2015regret}
A.~Koppel, F.~Y. Jakubiec, and A.~Ribeiro, ``Regret bounds of a distributed
  saddle point algorithm,'' in \emph{IEEE 40th International Conference on
  Acoustics, Speech and Signal Processing (ICASSP)}, 2015, pp. 2969--2973.

\bibitem{nedic2009approximate}
A.~Nedi{\'c} and A.~Ozdaglar, ``Approximate primal solutions and rate analysis
  for dual subgradient methods,'' \emph{SIAM Journal on Optimization}, vol.~19,
  no.~4, pp. 1757--1780, 2009.

\bibitem{bertsekas2009min}
D.~P. Bertsekas, ``Min common/max crossing duality: A geometric view of
  conjugacy in convex optimization,'' \emph{Lab. for Information and Decision
  Systems, MIT, Tech. Rep. Report LIDS-P-2796}, 2009.

\bibitem{bertsekas2015convex}
------, \emph{Convex Optimization Algorithms}.\hskip 1em plus 0.5em minus
  0.4em\relax Athena Scientific, 2015.

\bibitem{bertsekas1999nonlinear}
------, \emph{Nonlinear programming}.\hskip 1em plus 0.5em minus 0.4em\relax
  Athena scientific, 1999.

\bibitem{Alizadeh2015reduced}
M.~Alizadeh, A.~Scaglione, A.~Applebaum, G.~Kesidis, and K.~Levitt,
  ``Reduced-order load models for large populations of flexible appliances,''
  \emph{IEEE Transactions on Power Systems}, vol.~30, no.~4, pp. 1758--1774,
  2015.

\end{thebibliography}
\end{small}

\end{document}